\documentclass[11pt,a4paper]{article}\textwidth 6.5in\textheight 9in
\usepackage{amssymb}\usepackage[colorlinks]{hyperref}\usepackage{color}
	\usepackage{stmaryrd}\usepackage{mathrsfs}
	\usepackage{ulem}\usepackage{amsthm}
\begin{document}

\newcommand\hreff[1]{\href {http://#1} {\small http://#1}}
\newcommand\trm[1]{{\bf\em #1}} \newcommand\emm[1]{{\ensuremath{#1}}}

\newtheorem{thm}{Theorem} 
\newtheorem{lem}{Lemma}
\newtheorem{cor}{Corollary}
\newtheorem{con}{Conjecture} 
\newtheorem{prp}{Proposition}

\newtheorem{blk}{Block}
\newtheorem{dff}{Definition}
\newtheorem{asm}{Assumption}
\newtheorem{rmk}{Remark}
\newtheorem{clm}{Claim}
\newtheorem{example}{Example}

\newcommand{\ab}{a\!b}
\newcommand{\yx}{y\!x}
\newcommand{\yux}{y\!\underline{x}}

\newcommand\floor[1]{{\lfloor#1\rfloor}}\newcommand\ceil[1]{{\lceil#1\rceil}}

\newcommand\lea{\prec}\newcommand\gea{\succ}\newcommand\eqa{\asymp}
\newcommand\lel{\lesssim}\newcommand\gel{\gtrsim}

\newcommand\edf{{\,\stackrel{\mbox{\tiny def}}=\,}}
\newcommand\edl{{\,\stackrel{\mbox{\tiny def}}\leq\,}}
\newcommand\then{\Rightarrow}

\newcommand\km{{\mathbf {km}}}\renewcommand\t{{\mathbf {t}}}
\newcommand\KM{{\mathbf {KM}}}
\newcommand\md{{\mathbf {m}_{\mathbf{d}}}}\newcommand\mT{{\mathbf {m}_{\mathbf{T}}}}
\newcommand\KL{{\mathbf {KL}}}
\newcommand\Kd{{\mathbf{Kd}}} \newcommand\KT{{\mathbf{KT}}} 
\renewcommand\d{{\mathbf d}} \newcommand\w{{\mathbf w}}
\newcommand\Ks{\Lambda} \newcommand\q{{\mathbf q}}
\newcommand\E{{\mathbf E}} \newcommand\St{{\mathbf S}}
\newcommand\Q{{\mathbf Q}}
\newcommand\ch{{\mathcal H}} \renewcommand\l{\tau}
\newcommand\tb{{\mathbf t}} \renewcommand\L{{\mathbf L}}
\newcommand\bb{{\mathbf {bb}}}\newcommand\Km{{\mathbf {Km}}}
\renewcommand\q{{\mathbf q}}\newcommand\J{{\mathbf J}}
\newcommand\z{\mathbf{z}}
\newcommand\MLR{\mathrm{MLR}}
\newcommand\B{\mathbf{bb}}\newcommand\f{\mathbf{f}}
\newcommand\hd{\mathbf{0'}} \newcommand\T{{\mathbf T}}
\renewcommand\Q{\mathbb{Q}}
\newcommand\BT{\Sigma}
\newcommand\FIS{\BT^{*\infty}}\newcommand\C{\mathcal{L}}
\renewcommand\S{\mathcal{C}}\newcommand\ST{\mathcal{S}}
\newcommand\UM{\nu_0}\newcommand\EN{\mathcal{W}}
\renewcommand\em{\it}
\newcommand\lenum{\lbrack\!\lbrack}
\newcommand\renum{f\rbrack\!\rbrack}

\renewcommand\qed{\hfill\emm\square}

\title{\vspace*{-3pc} Randomness Conservation over Algorithms}

\author {Samuel Epstein\\samepst@bu.edu\\Boston University}

\maketitle

\newcommand\D{{\mathbf{D}}}
\newcommand\I{{\mathbf I}}
\newcommand\FS{\BT^*}
\newcommand\IS{\BT^\infty}
\newcommand\M{{\mathbf M}}
\newcommand\K{{\mathbf K}} 
\newcommand\m{{\mathbf {m}}}
\newcommand\R{\mathbb{R}}
\newcommand\N{\mathbb{N}}

\begin{abstract}
Current discrete randomness and information conservation inequalities are over total recursive functions, i.e. restricted to deterministic processing. This restriction implies that an \textit{algorithm} can break algorithmic randomness conservation inequalities. We address this issue by proving tight bounds of randomness and information conservation with respect to recursively enumerable transformations, i.e. processing by algorithms. We also show conservation of randomness of finite strings with respect to enumerable distributions, i.e. semicomputable semi-measures. 
\end{abstract}

\section{Introduction}

A finite string $x$ is typical of a computable probability measure $p$ if the length of its shortest description with respect to a prefixless universal algorithm, $\K(x)\,{\in}\,\N$, is close to the length of its $p$ code, $\ceil{-\log p(x)}$. $\K(x)$ is the Kolmogorov complexity of $x$. Such $p$-typical elements $x$ have a low deficiency of randomness, $\d_p(x) = \ceil{-\log p(x)}-\K(x)$. Atypical elements $x$ have extra regularity that allows them to be compressed to length $\K(x)\ll\ceil{-\log p(x)}$. The algorithmic joint information of individual finite strings $x$ and $y$ is $\I(x\,{:}\,y)=\K(x)\,{+}\,\K(y)\,{-}\,\K(x,y)$, the difference between the length of the shortest separate descriptions of $x$, $y$ and the length of the shortest joint description of $x$, $y$. Strings are independent if they have low mutual information. 

It has been shown that this definition of information is robust, i.e. invariant to processing by total functions $A$ over finite strings~\cite{Levin74}. Such deterministic functions cannot create an increase, $\left(\I(A(x)\,{:}\,y)\,{-}\,\I(x\,{:}\,y)\right)$, in the mutual information of strings $x$ and $y$ by more than a constant factor, dependent on $A$. Randomness is also conserved over total recursive functions, where $\d_{Ap}(A(x))$ is not much more than $\d_{p}(x)$. The measure $(Ap)(y)=\sum_{A(z)=y}p(z)$ is the image of $p$ under $A$. 

However randomness is not conserved over the set of limit computable functions $B$ over finite strings, i.e, functions computed by possibly non-halting algorithms. There exists a simple limit computable function $B$, a probability measure $p$, and a string $x$ such that (1) $x$ is $p$-typical and (2) $B(x)$ not $(Bp)$-typical (see theorem~\ref{noncons}). Thus randomness conservation inequalities do not hold with respect to $B$ and $x$. This article shows only {\em\textbf{exotic}} strings $x$ break information and randomness conservation inequalities over limit computable functions. We use $\I(x\,{;}\,\mathcal{H})\,{=}\,\K(x)\,{-}\,\K(x|\mathcal{H})$, to denote the amount of information the halting sequence $\mathcal{H}$ has about $x$. A string $x$ is called {\em\textbf{exotic}} if $\I(x\,{;}\,\mathcal{H})$ is large. We prove randomness and information conservation inequalities over non-exotic strings and limit computable functions. 

\newpage

In addition, this article uses notions of rarity over enumerable distributions. Enumerable distributions are semi-measures, nonnegative functions $p$ over strings such that $\sum_{x\neq\perp}p(x)\leq 1$. The images of measures under partial functions are semi-measures. This article also shows randomness and information conservation of infinite binary strings. This article contains continuous results limited to specialized definitions of information and rarity that are amenable to natural extensions of the proofs in Sections~\ref{sec:rand} and~\ref{sec:info}. 

\section{Conventions}
Let $\R$, $\Q$, $\N$, $\BT$, $\FS$, $\IS$ be the set of reals, rationals, natural numbers, bits, finite strings, and infinite binary sequences. $\FIS\edf \FS\cup\IS$. $(x0)^-=(x1)^- x$ for $x\in\FS$. The empty string is $\perp$. $\|x\|$ is the length of string $x$. $\alpha_{\leq n}$ represents the first $n$ bits of $\alpha\in\FIS$, and $\alpha_{> n}$ represents the remaining bits of $\alpha$. We say $x\sqsubseteq y$ if $x\in\FS$ is a prefix of $y\in\FIS$. $\langle x\rangle\in\FS$ is a self-delimiting representation of $x\in\FS$. We use $O_{a_1,\dots,a_k}(X)$ to denote any quantity bounded in magnitude by $CX$ where $C\in\N$ is dependent on parameters $a_1,\dots,a_k$. Let $[Z]=1$ if statement $Z$ is true, and $[Z]=0$, otherwise. Let $U_y$ be a fixed universal prefixless algorithm with auxilliary input $y\in\FIS$ and $\K(x|y)=\min\{\|p\|\,{:}\,U_y(p)=x\}$. 

Function $p:\FS\rightarrow\R_{\geq 0}$ is a semi-measure iff $\sum_{x\neq\perp}p(x)\leq 1$. $p$ is also a measure iff $\sum_{x\neq\perp}p(x)=1$.  A nonnegative function $f$ is semicomputable if the subgraph $\{(x,q)\,{:}\,f(x)>q\in\Q\}$ is enumerable. $f$ is $\alpha$-semicomputable if $f$ is semicomputable over $U_\alpha$. For a fixed enumeration of semicomputable semi-measures, $\{p_n\}\ni p$, $\K(p|y)=\min_{p_n=p}\K(n|y)$. The function $\m(x|y)$ is a majorant semicomputable semi-measure relativized to $y\in\FIS$.

By the coding lemma $\K(x|y)=-\log\m(x|y)\pm O(1)$.  Function $f:\FS\rightarrow\FS$ is a partial recursive function if it can be computed by a Turing machine $T$. In particular, $f(x)=\perp$ if $T$ halts without output or $T$ does not halt. A function $B:\FS\rightarrow\FS$ is limit computable if there is a Turing machine $T$ such that if $B(x)\neq\perp$ then machine $T$, when given input $x$, will at some point, print $B(x)$ to the output tape and make no further changes to the output (and then either halt or never halt). Note that if $B(x)=\perp$ then such $T$ will either (1) output nothing ($\perp$) when given input $x$, (and may either halt or not halt) or (2) it will never halt and continuously change the output tape. For a fixed enumeration $\{B_n\}\ni B$ of limit computable functions over strings, $\K(B)=\min_{B_n=B}\K(n)$. The halting sequence, $\mathcal{H}$, is the characteristic sequence of the domain of $U$. Chaitin's halting probability is defined by $\Omega = \sum_{x\neq\perp}\m(x)$. The deficiency of randomness of $x\in\FS$ with respect to an arbitrary semi-measure $p$, relative to $y\in\FIS$, is $\d_p(x|y)=\ceil{-\log p(x)}-\K(x|y)$. For semi-measure $p$, we say nonnegative function $t$ is a $p$-test, iff $\sum_{x\neq\perp} p(x)t(x)\leq 1$.  

\section{Related Work}

This work is resultant from my trip to Montpellier with Alexander Shen and P\'{e}ter G\'{a}cs. Kolmogorov complexity was introduced independently in ~\cite{Kolmogorov65, Solomonoff64,Chaitin75}. For a detailed history of Algorithmic Information Theory, we refer to~\cite{LiVi08}.~\cite{Levin74} introduced laws of information non-growth over deterministic functions and later revisited in~\cite{Levin84}. The definition of $\I^\infty(\alpha;\mathcal{H})$ and theorem~\ref{consasyminfo} relies on modified arguments of Section 2 in~\cite{Levin84}. An extension of rarity to semi-measures can be found in the recent work of \cite{Levin12} and also can seen in the work of~\cite{Levin84}. \cite{Gacs13} contains an extended survey of randomness conservation inequalities and also describes properties of the rarity term $\D$ used in this article.
\newpage
\section{Randomness Conservation}
\label{sec:rand}

The central trick of the article is using the fact that $2^{\I(x;\mathcal{H})}$ is a majorant  $\mathcal{H}$-semicomputable, $\m$-test. This enables proof techniques centered around the creation of $\m$-tests $t$. For any computable measure $p$, the function $\t_p(x)=\m(x)/p(x)$ is a majorant (up to a multiplicative constant) semicomputable $p$-test. For more information about universal semicomputable tests, see~\cite{Gacs13}. Proposition~\ref{prp1} follows from $\I(x;\mathcal{H})= \log \m(x|\mathcal{H})/\m(x)\pm O(1)=\d_\m(x|\mathcal{H})\pm O(1)$, and from the fact that $\m$ is computable from $\mathcal{H}$. 

\begin{prp} 
\label{prp1}
$\I(x;\mathcal{H}) = \d_\m(x|\mathcal{H})\pm O(1)$.
\end{prp}

Theorem~\ref{disccons} extends finite randomness conservation inequalites to limit computable functions $B$ and discrete semicomputable semi-measures $p$. For convenience we define $\d_p(\perp)=0$. Randomness is conserved for all strings that are non-exotic, i.e. have low mutual information with the halting sequence. The proof follows from the definition of an $\m$-test $t$ such that $\log t(x) = \d_{Bp}(B(x))- \d_p(x)\pm O_{B,p}(1)$. Theorem~\ref{noncons} shows the tightness of theorem~\ref{disccons}, and represents a generalization of the example used in the introduction. The proof of theorem~\ref{noncons}, leverages arguments in the proof of theorem~\ref{conttight}, adapted to the case of finite strings~\cite{BienvenuPoShHo13}.

\begin{thm}
\label{disccons}
For limit computable function $B:\FS\rightarrow\FIS$ and semicomputable semi-measure $p$, for all $x\in\FS$, $\d_{Bp}(B(x))< \d_p(x)+\I(x;\mathcal{H})+O_{B,p}(1)$.
\end{thm}
\begin{proof} We use the $\mathcal{H}$-semicomputable $\m$-test $t$, where $t(x)= \m(B(x))p(x)/(\m(x)Bp(B(x))$. Since $t$ is computable in the limit, it is $\mathcal{H}$-computable, with $\K(t|\mathcal{H})\,{=}\,O_{B,p}(1)$. Also $t$ is an $\m$-test with $\sum_{x}\m(x)t(x) = \sum_{x}\m(B(x))p(x)/Bp(B(x))$ $= \sum_{y\in\mathrm{Range}(B)}\m(y)\sum_{x:B(x)=y}p(x)/Bp(y)\leq \sum_{y}\m(y)\leq 1$. So $\d_{Bp}(B(x))- \d_p(x)<\log t(x) +O(1)< \d_\m(x|\mathcal{H})+\K(t|\mathcal{H}) <\I(x;\mathcal{H})+O_{B,p}(1)$. 
\end{proof}
\begin{thm} 
\label{noncons}
For all $b,c\,{\in}\,\N$, there exists limit computable function $B$, measure $p$, and string $x$ such that $\K(B,p)\,{=}\,O(\log bc)$, $\d_{Bp}(B(x))\,{=}\,b\,{+}\,c\,{\pm}O(\log bc)$, $\d_p(x)=b\,{\pm}O(\log bc)$, and $\I(x\,{;}\,\mathcal{H})\,{=}\,c\,{\pm}O(\log bc)$.
\end{thm}
\begin{proof}
Let $n=b+c$ and $\Omega_{\leq c}$ be a string representing the first $c$ bits of $\Omega$. So $\Omega_c$ is a random string and can be identified with a $c$-bit number that is enumerable from below.  Let $\lhd$ be a partial order over finite strings  where if $x,y\in\BT^n$, then $x \lhd y$ iff the $n$ bit number associated with $x$ is smaller than the $n$-bit number associated with $y$.

Let $B(x)=\perp$ if $\|x\|\neq n$. Otherwise $B(x)=x$ if $x_{>b} \lhd\Omega_c$, $B(x)=x_{\leq b}$ if $x_{>b}=\Omega_c$, or $B(x)=\perp$ if $\Omega_c\lhd x_{>b}$. $B$ can be enumerated by a non-halting Turing machine. Let $x=0^b\Omega_c$. $x$ has deficiency $b\pm O(\log bc)$ with respect to the uniform measure over $n$-bit strings $p(x)\,{=}\,[\|x\|\,{=}\,n]2^{-n}$. This is because $\d_p(x) = \ceil{-\log p(x)}-\K(x)=b+c-\K(\Omega_c)\pm O(\log (bc))=b\pm O(\log (bc))$. $B(x)$ has a greater randomness deficiency with respect to the probability measure $Bp$. This is because $\d_{Bp}(B(x)) = \d_{Bp}(0^b) = \ceil{-\log Bp(0^b)}-\K(0^b) = b+c-O(\log (bc))$. In addition, $\Omega_c$ is simple relative to $\mathcal{H}$ and $c$, since $\Omega$ can be computed to any degree of precision by an algorithm with access $\mathcal{H}$. Thus $\K(\Omega_c|\mathcal{H})=O(\log c)$ implies $\I(x;\mathcal{H})=\K(0^b\Omega_c)-\K(0^b\Omega_c|\mathcal{H})= c\pm O(\log (bc))$.
\end{proof}
\newpage
\section{Information Conservation}
\label{sec:info}
We prove information nongrowth over limit computable functions. Theorem~\ref{consinfo1} shows conservation of the symmetric information $\I(x\,{:}\,y)$ and theorem~\ref{consinfo2} shows conservation of asymmetric information $\I(x;\mathcal{H})$ between $x$ and the halting sequence.

\begin{thm}
\label{consinfo1}$ $\\
For $x,y\in\FS$ and limit computable function $B$,   $\I(B(x)\,{:}\,y) <\I(x\,{:}\,y)+\I(\langle x,y\rangle\,{;}\,\mathcal{H})+O_B(1)$.
\end{thm}
Let $t(x)= c2^{\I(B(x)\,{:}\,y)-\I(x\,{:}\,y)} =c\m(B(x),y)\m(x)/(\m(x,y)\m(B(x)))$ where $c$ is a constant solely dependent on $B$. $t$ is a $\m(x,y)$ test, with $\sum_{x,y}\m(x,y)t(x,y) = c\sum_{x,y}  \m(B(x),y)\m(x)/\m(B(x))=c\sum_{y,z}\sum_{x:B(x)=z} \m(z,y)\m(x)/\m(z)$. At this point, we can use the following inequality, where for all $z\in\FS$, $\m(z|\mathcal{H})>\sum_{x:B(x)=z}\m(B,x)/O(1)>\sum_{x:B(x)=z}\m(x)/O_B(1)$. So for proper choice of $c$, we have $\sum_{x,y}\m(x,y)t(x,y)=c\sum_{x,y}\m(B(x),y)\m(x)/\m(B(x))$
$=c\sum_{y,z}\sum_{x:B(x)=z} \m(z,y)\m(x)/\m(z)$ $< O_B(1)c\sum_{y,z}\m(z,y)\m(z|\mathcal{H})/\m(z){=}O_B(1)c\sum_{z}\sum_y\m(z,y)2^{\I(z;\mathcal{H})}$ $< \sum_z\m(z)2^{\I(z;\mathcal{H})}\leq O_B(1)c\leq 1$. Since $t$ is computable in the limit, $t$ is $\mathcal{H}$-semicomputable. So $\I(B(x)\,{:}\,y)-\I(x\,{:}\,y) <\log t(x,y)+O_B(1)< \d_{\m}(\langle x,y\rangle|\mathcal{H})+\K(\m,t|\mathcal{H}) +O_B(1) < \I(\langle x,y\rangle\,{;}\,\mathcal{H})+O_B(1)$. $\qed$

\begin{thm}
\label{consinfo2}
For partial recursive function $f$ and all $x\in\FS$, $\I(f(x);\mathcal{H}) < \I(x;\mathcal{H})+O_f(1)$.
\end{thm}
\begin{proof} 

We define the function $s:\FS\rightarrow \FS$ where $s(x) = c\m(f(x)|\mathcal{H})\m(x)/\m(f(x))$ when $f(x)\neq\perp$, and $s(x)=0$ otherwise. $c$ is a constant to be determined later. The function $s$ is a semi-measure by the following reasoning. Since $\m$ is a majorant semi-computable semi-measure, $\m(y)>\sum_{x:f(x)=y}\m(x)/O_f(1)$. So $\sum_x s(x) = c\sum_{y}\sum_{x:f(x)=y}\m(y|\mathcal{H})\m(x)/\m(y)<O_f(1)c\sum_y\m(y|\mathcal{H})\leq 1$, for proper choice of $c$ solely dependent on $f$. Since $s$ is computable relative to $\mathcal{H}$, we have that $\log s(x) < \log \m(x|\mathcal{H}) + \K(s)+O(1)$. So $\I(f(x);\mathcal{H}) < \I(x;\mathcal{H})+O_f(1)$.
\end{proof}

\section{Continuous Conservation}
\label{seccont}
Some care is needed to extend the asymmetric information term $\I(x;\mathcal{H})$ to the case of infinite sequences. For $x\in\FS$,  $\Gamma_x\subseteq\IS$ represents the set of all infinite strings $\alpha\in\IS$ where  $\alpha \sqsupseteq x$. Thus $\IS$ is a Cantor space and the set of intervals, $\{\Gamma_x\,{:}\,x\in\FS\}$, is a binary topological basis for $\IS$. Continuous semi-measures $P$ are defined using functions $P:\FS\rightarrow\R_{\geq 0}$ such that $P(\perp)\leq 1$ and $P(x)\geq P(x0)+P(x1)$. We extend $P$ to $\IS$, with $P(\Gamma_x)=P(x)$ and for any open set $U\subseteq\IS$, $P(U) = \sum_{\Gamma_a\subseteq U}P(\Gamma_a)$, where $\Gamma_a$ are the maximal intervals of $U$.  For any set $D\subset\FS$ of finite strings, $P(D)=P(\cup\{\Gamma_x\,{:} \,{x}\,{\in}D\})$. 

Let $\{P_i\}$ be an enumeration of all semicomputable continuous semi-measures. We use the fixed majorant semicomputable continuous semi-measure, $\M(x) = \sum_i2^{-i}P_i(x)$. Semicontinuous functions $f:\IS\rightarrow \R_{\geq 0}\cup\{\infty\}$ are defined with respect to their elementary functions $f^\Gamma:\FS\rightarrow \R_{\geq 0}$, with $f(\alpha)= \sup_{x\sqsubseteq\alpha}f^\Gamma(x)$. Such $f$ is semicomputable if its elementary function $f^\Gamma$ is semicomputable. Let $\{f^\Gamma_n\}\ni f^\Gamma$ be a fixed enumeration such elementary functions and let $\K(f) = \min_{f^\Gamma_n=f^\Gamma}\K(n)$. For continuous semi-measure $Q$, we say semicontinuous $t$ is a $Q$-test if $Q(\{\alpha:t(\alpha)> 2^m\})< 2^{-m}$ for all $m\in\N$. The domain of such $t$ is extended to finite strings $x$, with $t(x)=\min_{x\sqsubset\alpha}t(\alpha)$. 

The function \mbox{$\sqsubseteq$-$\sup$} is the supremum under the partial order of $\sqsubseteq$ on $\FIS$. A function $\nu\,{:}\,\FS\,{\rightarrow}\,\FS$ is monotone iff for all $p,q\in\FS$,  $\nu(p)\,{\sqsubseteq}\,\nu(pq)$. Then monotone function $B\,{:}\,\FIS\,{\rightarrow}\,\FIS$ denotes the unique extension of $\nu$, where $B(p)\edf \sqsubseteq$-$\sup\, \{\nu(p_{\leq n})\,{:}\,n\,{\leq}\,\|p\|,n\,{\in}\,\N\}$ for all $p\,{\in}\,\FIS$. 
We say $\overline{\nu}$ is a recursive monotone function if there is a Turing machine $T$ with a write only output tape such that $\nu(x)$ is equal to the output of $T$ on input $x$. For $x\in\FS$, we say $\overline{\nu}^{-1}(x)=D\subseteq \FS$ is the prefix-free set of finite strings $y$ such that $\nu(y^{-})\sqsubset x \sqsubseteq \nu(y)$.

Let $\mathcal{T}_{\M'}$ be an enumeration of all $\mathcal{H}$-semicomputable $\M$-tests. The information that $\mathcal{H}$ has about  $\alpha\in\IS$ is defined to be the logarithm of a weighted sum of such tests, with $\I^\infty(\alpha;\mathcal{H})=\log \sum_{t_i\in\mathcal{T}_{\M'}} \m(i|\mathcal{H})t_i(\alpha)$. Note however that $2^{\I^\infty(\alpha;\mathcal{H})}$ is not necessarily an $\M$-test, since $\M$ is superadditive. This is a major difference from the finite case, where $2^{\I(x;\mathcal{H})}$ is an $\m$-test. The domain of $\I^\infty(\alpha;\mathcal{H})$ is extended to $\FS$ with $\I^\infty(x;\mathcal{H})=\inf_{x\sqsubseteq \alpha\in\IS}\I^\infty(\alpha;\mathcal{H})$ for all $x\in\FS$. Theorem~\ref{consasyminfo} represents the continuous variant of theorem~\ref{consinfo2}. The theorem shows the asymmetric information term $\I^ \infty$ has nongrowth properties with respect to recursive monotone transformations.

\begin{thm}
\label{consasyminfo}
For recursive monotone function $B:\FIS\rightarrow\FIS$ and all $\alpha\in\IS$, $\I^\infty({B}(\alpha)\,{;}\,\mathcal{H}) < \I^\infty(\alpha\,{:}\,\mathcal{H})+O_{B}(1)$.\end{thm}
\begin{proof} 
Let $\mu(x) = \M({B}^{-1}(x))$.
$\mu$ is a semi-measure because $\mu(\perp)\leq 1$ and 
$\mu(x0)+\mu(x1)$
$=\M({B}^{-1}(x0))+\M({B}^{-1}(x1)) \leq \M({B}^{-1}(x))= \mu(x).$ $\mu$ is semicomputable because $\M$ is semicomputable and ${B}^{-1}(x)$ is enumerable. Since $\M$ is a majorant semicomputable semi-measure, there exists $c$ solely dependent on $B$ with $\mu(x) \leq c\M(x)$. So for any open set $U$, $\mu(U)< c\M(U)$. 

For all $t\in \mathcal{T}_{\M'}$ and $m\,{\in}\,\N$, let $S$ be a set of finite strings  representing the maximal binary intervals of $\{\alpha\,{:}\,t(\alpha)\,{>}\,2^m\}$. So $2^{-m}> \M(S) \geq \mu(S)/c$. So $c2^{-m}> \mu(S)=\M(B^{-1}(S))$
$\geq\M(\{\alpha\,{:}\,Bt(\alpha)> 2^m\})$, where $Bt(\alpha) = \min_{\beta\sqsupseteq B(\alpha)}t(\beta)$.  So $Bt(\alpha)/c\in \mathcal{T}_{\M'}$ is an $\M$ test. $Bt$ is $\mathcal{H}$-computable because $t$ and $B$ are $\mathcal{H}$-semicomputable, with $\K(Bt|\mathcal{H}) < \K(t|\mathcal{H})+O_B(1)$. So this implies the inequality $\I^\infty(B(\alpha);\mathcal{H})=\log\sum_{t\in \mathcal{T}_{\M'}}\m(t|\mathcal{H}) t(B(\alpha))$ $= \log\sum_{t\in \mathcal{T}_{\M'}}\m(t|\mathcal{H})Bt(\alpha)<$ $\log\sup_{Bt\in \mathcal{T}_{\M'}}\m(Bt|\mathcal{H})Bt(\alpha)+O_B(1)<$ $\log\sum_{t\in \mathcal{T}_{\M'}}\m(t|\mathcal{H}) t(\alpha)+O_B(1)= \I^\infty(\alpha;\mathcal{H})+O_B(1)$. 
\end{proof}

For continuous measures $P$ and finite strings $x$, we define the following finite deficiency function $\D_P(x) = \log \max_{y\sqsubseteq x} \left(\M(y)/P(y)\right)$. Its extension to infinite strings $\alpha\in\IS$ is denoted by $\D_P(\alpha) = \sup_n \D_P(\alpha_{\leq n})$. $\D_P$ is a $P$-test, with for all $m\in\N$, $P(\{\alpha\,{:}\,\D_P(\alpha)>2^m\})<2^{-m}$. $\D$ is also a probability bounded ML test~\cite{Gacs13}. For continuous semi-measures $P$, $\D_P(\alpha)=\sup_n \D_P(\alpha_{\leq n})$ is the largest, non-increasing on $P$, semicontinuous on $\alpha$, extension of $\D$. This term $\D$ is admittedly not a definitive definition of randomness with respect to a continuous (semi)measure. However proving properties about $\D$ could have utility in future applications. Theorem~\ref{contrandcons} follows using the same general proof technique as theorem~\ref{disccons}, with the construction of a $\M$-test, $t$. Theorem~\ref{conttight} shows tightness of theorem~\ref{contrandcons} (on a finite/infinite level).

\begin{thm}
\label{contrandcons}
For recursive monotone function $B:\FIS\rightarrow\FIS$ and continuous semicomputable semi-measure $P$, for all $\alpha\in\IS$, $\D_{BP}(B(\alpha))<\D_P(\alpha)+\I^\infty(\alpha\,{;}\,\mathcal{H})+O_{B,P}(1)$.
\end{thm}
\begin{proof}
We the semicontinuous function $t(\alpha)=\sup_n P(\alpha_{\leq n})\M(B(\alpha_{\leq n})/(\M(\alpha_{\leq n})BP(B(\alpha_{\leq n})))$. $t(\alpha)$ is a $\M$ test. Indeed, let $m\in\N$ and $S=\{\alpha\,{:}\,t(\alpha){>} 2^m\}$. So $\M(S)\leq\sum_{x\in S}2^{-m}\M(x)t(x)$ $=2^{-m}\sum_{x\in S}P(x)\M(B(x))/BP(B(x))\leq 2^{-m}\sum_{y\in B(S)}\M(y)\sum_{x,B(x)=y}P(x)/BP(y)$ $=2^{-m}\M(B(S))$ $\leq 2^{-m}$. Since $t$ is computable in the limit, $t$ is $\mathcal{H}$-semicomputable, with $\K(t|\mathcal{H})=O_{B,P}(1)$. Since $t\in\mathcal{T}_{\M'}$ is an $\M$ test, $\D_{BP}(B(\alpha))-\D_P(\alpha)<\log t(\alpha) < \I^\infty(\alpha;\mathcal{H})+\K(t|\mathcal{H}) < \I^\infty(\alpha;\mathcal{H})+O_{B,P}(1)$. 
\end{proof}

\begin{thm}[\cite{BienvenuPoShHo13}] 
\label{conttight}
There exists limit computable $B\,{:}\,\IS\,{\rightarrow}\,\FIS$, continuous measure $P$, and infinite sequence $\alpha\,{\in}\,\IS$, where $\D_P(\alpha){<}\infty$, 
$\D_{B(P)}(B\alpha){=}\infty$, and $\I(\alpha;\mathcal{H}){=}\infty$. 
\end{thm}
\section{Acknowledgements}
I would like to thank Laurent Bienvenu, P\'{e}ter G\'{a}cs, Wolfgang Merkle, Joseph Miller, Chris Porter, Paul Shafer,  and Alexander Shen for insightful discussions and reference material.

\end{document}